\documentclass[12pt, a4paper]{article}
\usepackage[utf8]{inputenc}
	 \usepackage{amsmath,amsthm}
	 \usepackage{amsfonts, amssymb}
 \usepackage{float}
 \usepackage{url}
\usepackage{graphicx}
 \usepackage{xcolor}
\usepackage{enumerate}
\usepackage{enumitem}
\usepackage{array}
\usepackage[urlcolor=black,linkcolor=black,citecolor=black,colorlinks=true]{hyperref}
\usepackage[left=2.9cm, right=2.9cm, top=2.25cm, bottom=3.25cm]{geometry}
\usepackage{pict2e}
\usepackage{placeins}
\usepackage{tikz}
\usepackage{subcaption}
\usepackage{comment}

\numberwithin{equation}{section}

\newtheorem{theoreme}{Theorem}[section]
\newtheorem{lemme}[theoreme]{Lemma}

\newtheorem{proposition}[theoreme]{Proposition}
\newtheorem{corollaire}[theoreme]{Corollary}
\newtheorem{exemple}[theoreme]{Example}
\newtheorem{remarque}[theoreme]{Remark}
\newtheorem{hypothese}[theoreme]{Assumption}

\title{Continuous-time modeling and bootstrap for Schnieper's reserving}

\author{Nicolas Baradel\footnote{Inria, CMAP, CNRS, \'{E}cole polytechnique, Institut Polytechnique de Paris, 91200 Palaiseau, nicolas.baradel@polytechnique.edu.}}

\begin{document}

\maketitle

\begin{abstract}
We revisit Schnieper's model, which decomposes incurred but not reported (IBNR) reserves into two components: reserves for newly reported claims (true IBNR) and reserves for changes over time in the estimated cost of already reported claims (IBNER). We propose a continuous-time stochastic model for the aggregate claims process, driven by a random Poisson measure for the arrival of newly reported claims and by Brownian motion for the cost fluctuations of reported claims. This framework is consistent with the key assumptions of Schnieper's original approach. Within this setting, we develop a bootstrap method to estimate the full predictive distribution of claims reserves. Our approach naturally accounts for asymmetry, ensures non-negativity, and respects intrinsic bounds on reserves, without requiring additional assumptions. We illustrate the method through a case study and compare it with alternative stochastic techniques based on Schnieper's model.
\end{abstract}

\section{Introduction}

In the context of excess-of-loss reinsurance, \cite{schnieper1991separating} proposed a model that decomposes incurred but not reported (IBNR) reserves into two components: \textit{true IBNR}, representing claims that have not yet been reported, and \textit{IBNER} (incurred but not enough reported), capturing changes over time in the estimated ultimate cost of already reported claims.

\medbreak

Although \cite{mack1993distribution} incorporated some elements of Schnieper's approach, the model has received relatively little attention in the subsequent literature. Notable extensions include \cite{liu2009predictive}, which derives the conditional mean squared error of prediction (MSEP) as a measure of reserve uncertainty; \cite{liu2009bootstrap}, which develops a non-parametric bootstrap procedure to estimate the predictive distribution of reserves; and \cite{baradel2025modeling}, which studies and extends the special case of claim counts exceeding a priority threshold, originally introduced in \cite{schnieper1991separating}.

\medbreak

A parallel strand of research has explored continuous-time stochastic models for loss reserving. For instance, \cite{baradel2025ime} proposes a continuous-time formulation of Mack's model using stochastic differential equations driven by Brownian motion and develops a bootstrap method to estimate the conditional MSEP and quantiles. This framework provides a coherent parametric structure with independent increments and avoids the generation of negative values in cumulative processes during bootstrapping.

\medbreak

In this paper, we extend the same continuous-time philosophy to Schnieper's model. We model the \textit{true IBNR} component, which is non-decreasing, using a random Poisson measure to drive the arrival of newly reported claims, while the \textit{IBNER} component is driven by Brownian motion, following \cite{baradel2025ime}. The key advantage of this approach is that it enables parametric bootstrap simulation of the total claims reserves. This naturally incorporates asymmetry, enforces non-negativity, and respects intrinsic bounds on reserves, all within Schnieper's original decomposition and without requiring residual-based adjustments or additional assumptions.

\medbreak 

Other continuous-time frameworks for loss reserving have primarily focused on the chain-ladder method. For example, \cite{miranda2013continuous} develops a bivariate density approach with kernel-based estimation to model individual claim development within the triangle, while \cite{bischofberger2020continuous} employs marked point processes to capture the timing and magnitude of claim payments. In contrast, our work operates on aggregated data and introduces a continuous-time stochastic process that aligns with \cite{baradel2025ime} while faithfully respecting Schnieper's key assumptions.

\medbreak 

The remainder of the paper is organized as follows. Section 2 reviews Schnieper's general model and its key estimators. Section 3 introduces our continuous-time stochastic model for the aggregate claims process, derives its main properties, and establishes its consistency with Schnieper's framework. Section 4 presents a tailored bootstrap procedure that accounts for parameter uncertainty. Finally, Section 5 illustrates the method through a case study and compares it with alternative stochastic approaches based on Schnieper's model.

\section{Schnieper's model}

Schnieper's model \cite{schnieper1991separating} provides a probabilistic framework that separates two different behaviors in the IBNR data: 

\begin{itemize}
    \item The occurrence of newly reported claims, assumed to occur randomly in proportion to the exposure ;
    \item The development (re-estimation) of the ultimate cost of claims already reported in previous development periods.
\end{itemize}

\medbreak

The model introduces the two processes :
\begin{itemize}
    \item $(N_{i,j})_{1 \leq i, j \leq n}$ which represents the total amount of new excess claims, referring to claims that have not been recorded as excess claims in previous development years,
    \item $(D_{i,j})_{1 \leq i, j \leq n}$ which represents the decrease in the total claims amount between development years $j-1$ and $j$, concerning claims that were already known in development year $j-1$.
\end{itemize}

Note that $(D_{i,j})_{1 \leq i, j \leq n}$ may be negative (reflecting increases in estimated ultimate cost) and by definition, $D_{i, 1} = 0$ for all $1 \leq i \leq n$.

\bigskip

Given these processes, the cumulative incurred data $(C_{i,j})_{1 \leq i, j \leq n}$ are obtained recursively as follows:

\begin{equation}
    \begin{aligned}
    &C_{i, 1} = N_{i, 1}, & & 1 \leq i \leq n \\
    &C_{i,j+1} = C_{i,j} + N_{i,j+1} - D_{i,j+1}, & & 1 \leq i \leq n, \ 1 \leq j \leq n-1
    \end{aligned}
\end{equation}
We also assume the exposures $(E_i)_{1 \leq i \leq n}$ are known and non-negative. For each accident year $1 \leq i \leq n$, define the filtration :
    \[
        \mathcal{F}^i_k := \sigma\left(N_{i, j}, D_{i, j}, j \leq k\right), \ \ 1 \leq k \leq n.
    \]
Schnieper's general model relies on the following assumptions:
    \begin{hypothese}\label{H_schnieper} 
        \leavevmode\begin{itemize}
            \item[H1] The random variables $(N_{i_{1}, j}, D_{i_{1}, j})_{1 \leq j \leq n}$ and $(N_{i_{2}, j}, D_{i_{2}, j})_{1 \leq j \leq n}$ are independent for $i_{1} \not= i_{2}$. 
            \item[H2] For $1 \leq j \leq n$, there exists $\Lambda_{j} \geq 0$ and, for $1 \leq j \leq n - 1$, there exists $\Delta_j \leq 1$ such that
                \[
                    \begin{aligned}
                    \mathbb{E}(N_{i, j} \mid \mathcal{F}_{j}^i) &= \Lambda_{j}E_{i}, \quad &1 \leq i \leq n,\\
                    \mathbb{E}(D_{i, j+1} \mid \mathcal{F}_{j}^i) &= \Delta_{j}C_{i,j}, \quad &1 \leq i \leq n.
                    \end{aligned}
                \]
            \item[H3] For $1 \leq j \leq n$, there exists $\Sigma_{j} \geq 0$ and, for $1 \leq j \leq n - 1$, there exists $T_{j} \geq 0$ such that
                \[
                    \begin{aligned}
                    Var(N_{i, j} \mid \mathcal{F}_{j}^i) &= \Sigma_{j}^{2}E_{i}, \quad &1 \leq i \leq n,\\
                    Var(D_{i, j+1} \mid \mathcal{F}_{j}^i) &= T_{j}^{2}C_{i,j}, \quad &1 \leq i \leq n.
                    \end{aligned}
                \]
        \end{itemize}
    \end{hypothese}
    
Assuming, in addition, that $N_{i,j}$ and $D_{i,j}$ are independent for all $1 \leq i,j \leq n$, we can derive explicit expressions for the conditional mean and variance of $C_{i,j}$ given $\mathcal{F}_s^i$:
\begin{lemme}\label{schniepercor}
For all  $1 \leq i \leq n$ and $s \leq j \leq n$,
\[
        \begin{aligned}
            &\mathbb{E}(C_{i,j} \mid \mathcal{F}_{s}^i) = \left(\prod_{k = s}^{j-1}(1-\Delta_{k})\right)C_{i,s} + E_{i}\sum_{k = s}^{j}\Lambda_{k}\left(\prod_{\ell = k}^{j-1}(1-\Delta_{\ell})\right),\\
            &Var(C_{i,j} \mid \mathcal{F}_{s}^i) = \\
            &\sum_{k=s}^{j-1}\left[\left(\prod_{\ell=k+1}^{j-1}(1-\Delta_\ell)^2\right)\left[\Sigma_{k+1}^{2}E_i + T_{k}^{2}\left(\left[\prod_{\ell=s}^{k-1}(1-\Delta_\ell)\right]C_{i, s} + E_i\sum_{\ell=s}^{k}\Lambda_\ell\prod_{m = \ell}^{k-1}(1-\Delta_m)\right)\right]\right].
        \end{aligned}
        \]
    \end{lemme}
Schnieper provides the following unbiased estimators for the development parameters and variances:
\begin{equation}\label{est_fs}
    \begin{aligned}
    \widehat{\Lambda}_{j} &:= \frac{\sum_{i = 1}^{n - j+1}N_{i, j}}{\sum_{i = 1}^{n - j+1}E_{i}},  &1 \leq j \leq n,\\
    \widehat{\Delta}_{j} &:= \frac{\sum_{i = 1}^{n - j}D_{i, j+1}}{\sum_{i = 1}^{n - j}C_{i,j}},  &1 \leq j \leq n-1,\\
    \widehat{\Sigma}_{j}^2 &:= \frac{1}{n-j}\sum_{i = 1}^{n-j+1}E_{i}\left(\frac{N_{i,j}}{E_{i}} - \widehat{\Lambda}_{j}\right)^{2}, &1 \leq j \leq n-1,\\
    \widehat{T}_{j}^2 &:= \frac{1}{n-j-1}\sum_{i = 1}^{n-j}C_{i,j}\left(\frac{D_{i,j+1}}{C_{i,j}} - \widehat{\Delta}_{j}\right)^{2}, &1 \leq j \leq n-2.
    \end{aligned}
\end{equation}
For $j = n-1$, Schnieper sets $\widehat{\Sigma}_{n-1}^2 = \widehat{T}_{n-1}^2 = 0$.

Analogous to Mack's chain-ladder method, we can deduce an unbiased estimator for the ultimate value by observing that the estimators in the products above are conditionally uncorrelated:
\begin{equation*}
	\widehat{C}_{i,n} := \left(\prod_{k = n+1-i}^{n-1}(1-\widehat{\Delta}_{k})\right)C_{i,n+1-i} + E_{i}\sum_{k = n+2-i}^{n}\widehat{\Lambda}_{k}\left(\prod_{\ell = k}^{n-1}(1-\widehat{\Delta}_{\ell})\right), \quad 2 \leq i \leq n,
\end{equation*}
The corresponding estimator of the total outstanding reserve is then
\begin{equation*}
	\widehat{R} := \sum_{i=2}^{n} \widehat{C}_{i,n} - C_{i, n-i+1}.
\end{equation*}

Furthermore, \cite{liu2009predictive} derives an explicit formula for the MSEP under Schnieper's model.

\section{A continuous-time model}

Let $(\Omega, \mathcal{F}, \mathbb{P})$ be a probability space sufficiently rich to support an $n$-dimensional Brownian motion $W = (W^i)_{1 \leq i \leq n}$ and  $n$ independent Poisson random measures $\Pi^i(dt, dz)$ on $[0, n] \times \mathbb{R}_{+}^{*}$, each with intensity measure $\lambda_{t}E_i \, dt \otimes \nu(dz)$, where $\nu$ is a finite measure on $\mathbb{R}_{+}^{*}$ satisfying $\int_{\mathbb{R}_{+}^{*}} z^{2} \, \nu(dz) < \infty$ and $(E_i)_{1 \leq i \leq n}$ are the exposures.

\medbreak

We define the following filtrations, which represent the information available at development time $t$ for accident year $i$:
    \[
        \mathcal{F}_{t}^i := \sigma(\Pi^{i}([0, s], \cdot), \ W_{s}^i, \ s \leq t), \quad 1 \leq i \leq n, \ 0 \leq t \leq n.
    \]
We define the filtration of the entire knowledge at time $t \in [0, 2n]$.
    \[
        \mathcal{F}_{t} := \sigma(\Pi^{i}([0, s], \cdot), \ W_{s}^i, \ i+s \leq t+1, s \leq n), \quad 0 \leq t \leq 2n.
    \]
For each $1 \leq i \leq n$, let $(C_{t}^{i})_{0 \leq t \leq n}$ denote the cumulative incurred claims process for accident year $i$.
\begin{hypothese}
        \leavevmode
	There exist three bounded measurable functions $\lambda : [0, n] \rightarrow \mathbb{R}_+$, $\delta : [0, n] \rightarrow \mathbb{R}$ and $\tau : [0, n] \rightarrow \mathbb{R}_{+}$ such that, for all $1 \leq i \leq n$, $(C_{t}^i)_{0 \leq t \leq n}$ is the unique strong solution of the stochastic differential equation:
\begin{equation} \label{EDS_C}  
    C_{t}^i = C_{s}^i + \int_{s}^{t}\int_{\mathbb{R}_+^*} z \Pi^i(du, dz) - \int_{s}^{t} \delta_{u} C_u^idu + \int_{s}^{t} \tau_{u} \sqrt{C_u^i}dW_u^i, \quad 0 \leq s \leq t \leq n,
\end{equation}
with initial condition $C_{0}^i = 0$. Moreover, 
\begin{equation}\label{EDS_L2}
	\mathbb{E}\left[\sup_{0 \leq t \leq n} (C_{t}^i)^2\right] < +\infty, \quad 1 \leq i \leq n.
\end{equation}
\end{hypothese}
The jump term in \eqref{EDS_C} can be written explicitly as, for all $1 \leq i \leq n$,
\[
    \int_{s}^{t}\int_{\mathbb{R}_+^*} z \Pi^i(du, dz) = \sum_{\ell=\Pi_{s}^i+1}^{\Pi_t^i}Z_{\ell}^{i}, \quad 0 \leq s \leq t \leq n,
\]
where $(\Pi^i_t)_{t \geq 0} := (\int_{0}^{t}\int_{\mathbb{R}_+^*}\Pi^i(du, dz))_{t \geq 0}$ is an inhomogeneous Poisson process with intensity $\lambda_{t}E_i$ and the sizes $(Z_{k}^i)_{k\geq 0} \overset{i.i.d.}{\sim} \mathbb{P}_{Z}$ having finite second moment (independent of the Poisson process).

\medbreak

The stochastic differential equation \eqref{EDS_C} admits a unique non-negative strong solution because on each inter-jump interval, as established in, for instance, \cite{yamada1971uniqueness} or \cite[Theorem 4.6.11]{meleard2016modeles}, and the global solution is then constructed iteratively by piecing these segments together at the jump times. For \eqref{EDS_L2}, the result is known in the fully Lipschitz case, however, the same standard proof still works here (from \eqref{EDS_C}, use of Jensen and Doob's maximal inequality, and we conclude with Gronwall lemma).

\medbreak

Note that \eqref{EDS_C} involves only three time-dependent coefficients, with no analogue to a volatility term $\Sigma$. We now derive the first two conditional moments of the processes $C^i$ in order to verify consistency with Assumption~\ref{H_schnieper}.

\begin{proposition}\label{Cmoments}
The first two conditional moments of the processes $(C_{t}^i)_{0 \leq t \leq n}$ are, for all $0 \leq s \leq t \leq n$ and $1 \leq i \leq n$,
\[
    \begin{aligned}
        \mathbb{E}\left(C_{t}^i \mid \mathcal{F}_s^i\right) = \, &C_{s}^i e^{\int_{s}^{t} -\delta_u du } + \mathbb{E}(Z)E_i\int_{s}^{t}\lambda_ue^{\int_{u}^{t} -\delta_v dv }du,\\
        Var\left(C_{t}^i \mid \mathcal{F}_s^i\right) = \, &C_s^i\int_{s}^{t}\tau_u^2e^{-\int_{s}^u \delta_z dz - \int_{u}^t 2\delta_z dz}du \\
    &+ \mathbb{E}(Z)E_i\int_{s}^{t}\tau_u^2e^{-\int_{u}^t 2\delta_v dv}\int_{s}^{u}\lambda_v e^{\int_{v}^{u}-\delta_y dy}dvdu \\
    &+ \mathbb{E}(Z^2)E_i\int_{s}^{t}\lambda_ue^{-2\int_{u}^{t}\delta_v dv}du.
    \end{aligned}
\]
\end{proposition}
\begin{proof}
Fix $i \in \{1, \ldots, n\}$. For convenience, denote $C^i$ as $C$, $\mathcal{F}^i$ for $\mathcal{F}$, $\Pi^i$ as $\Pi$, and $W^i$ as $W$ throughout this proof.\\
\textbf{1.} Applying the expected value operator $\mathbb{E}$ to (\ref{EDS_C}) and utilizing \eqref{EDS_L2} for the local martingale yields:
\[
    \mathbb{E}\left(C_{t} \mid \mathcal{F}_s\right) = C_{s} + \mathbb{E}(Z)E_i\int_{s}^{t}\lambda_udu - \int_{s}^{t} \delta_{u} \mathbb{E}\left(C_{u} \mid \mathcal{F}_s\right)du. 
\]
This forms a simple linear inhomogeneous ordinary differential equation with the unique solution:
\begin{equation}\label{esperance}
    \mathbb{E}\left(C_{t} \mid \mathcal{F}_s\right) = C_{s} e^{\int_{s}^{t} -\delta_u du } + \mathbb{E}(Z)E_i\int_{s}^{t}\lambda_ue^{\int_{u}^{t} -\delta_v dv }du.
\end{equation}
\textbf{2.} It\^{o}'s formula gives:
\[\begin{aligned}
    C_{t}^2 &= C_{s}^2 + 2\int_{s}^{t}  -\delta_u C_u^2du + 2\int_{s}^{t}  \tau_u C_u\sqrt{C_u}dW_u + \int_{s}^{t} \tau_u^2 C_udu,\\
    &+ \int_{s}^{t}\int_{\mathbb{R}_+^{*}}\left(2C_{u-}z + z^2\right)\Pi(du, dz)
    \end{aligned}
\]
Introducing the stopping times $T_{m} := \inf\{t \geq s: C_{t} = m\}$, which tends to infinity a.s. as $m \rightarrow +\infty$, we apply the expected value operator:

\begin{equation}\label{EC2_TM}
    \begin{aligned}        
        \mathbb{E}\left(C_{t}^2 \mid \mathcal{F}_s\right) &= C_{s}^2 - 2\mathbb{E}\left(\int_{s}^{t\wedge T_m}  \delta_u C_u^2  du \mid \mathcal{F}_s\right) + \mathbb{E}\left(\int_{s}^{t\wedge T_m} \tau_u^2 C_u du\mid \mathcal{F}_s\right) \\
        &+ \mathbb{E}\left(\int_{s}^{t\wedge T_m} 2C_{u-} z + z^2 \Pi(du, dz)\mid \mathcal{F}_s\right)
    \end{aligned}
\end{equation}

Taking the limit as $m\rightarrow +\infty$, and using \eqref{EDS_L2} along with the dominated convergence theorem, we obtain:

\begin{equation}\label{EC2}
    \begin{aligned}        
        \mathbb{E}\left(C_{t}^2 \mid \mathcal{F}_s\right) &= C_{s}^2 - 2\int_{s}^{t}  \delta_u \mathbb{E}\left(C_u^2 \mid \mathcal{F}_s\right) du + \int_{s}^{t} \tau_u^2 \mathbb{E}\left(C_u \mid \mathcal{F}_s\right) du \\
        &+ 2 \mathbb{E}(Z)E_i \int_{s}^{t}  \lambda_u \mathbb{E}\left(C_u \mid \mathcal{F}_s\right)du + \mathbb{E}(Z^2)E_i\int_{s}^{t}\lambda_udu
    \end{aligned}
\end{equation}

On the other hand, $t \mapsto\mathbb{E}\left(C_{t} \mid \mathcal{F}_s\right)^2$ satisfies the following differential equation: 
\[
    \mathbb{E}\left(C_{t} \mid \mathcal{F}_s\right)^2 = C_{s}^2 - 2\int_{s}^{t} \delta_{u} \mathbb{E}\left(C_{u} \mid \mathcal{F}_s\right)^2du + 2\mathbb{E}(Z)E_i \int_{s}^{t} \lambda_u \mathbb{E}\left(C_{u} \mid \mathcal{F}_s\right)du.
\]
Combining it with (\ref{EC2}) and (\ref{esperance}) gives:
\[
    Var\left(C_{t} \mid \mathcal{F}_s\right) =  -2\int_{s}^{t}  \delta_u Var(C_u \mid \mathcal{F}_s)du + \int_{s}^{t} \tau_u^2 \mathbb{E}\left(C_u \mid \mathcal{F}_s\right) du  + \mathbb{E}(Z^2)E_i\int_{s}^{t}\lambda_udu.
\]
\[
    \begin{aligned}
    Var\left(C_{t} \mid \mathcal{F}_s\right) &=  -2\int_{s}^{t}  \delta_u Var(C_u \mid \mathcal{F}_s)du + C_{s} \int_{s}^{t} \tau_u^2 e^{\int_{s}^{u} -\delta_v dv } du \\
    &\mathbb{E}(Z)E_i\int_{s}^{t} \tau_u^2 \int_{s}^{u}\lambda_ve^{\int_{v}^{u} -\delta_y dy }dv du + \mathbb{E}(Z^2)E_i\int_{s}^{t}\lambda_udu.
    \end{aligned}
\]
It is a linear non-homogeneous ordinary differential equation whose unique solution is:
\[
    \begin{aligned}
    Var\left(C_{t} \mid \mathcal{F}_s\right) = \, &C_s\int_{s}^{t}\tau_u^2e^{-\int_{s}^u \delta_v dv - \int_{u}^t 2\delta_v dv}du \\
    &+ \mathbb{E}(Z)E_i\int_{s}^{t}\tau_u^2e^{-\int_{u}^t 2\delta_v dv}\int_{s}^{u}\lambda_v e^{\int_{v}^{u}-\delta_y dy}dvdu \\
    &+ \mathbb{E}(Z^2)E_i\int_{s}^{t}\lambda_ue^{-\int_{u}^{t}2\delta_v dv}du.
    \end{aligned}
\]

\end{proof}
We have so far defined the cumulative claims process $(C_t^i)_{0 \leq t \leq n}$. However, to connect our continuous-time model to Schnieper's discrete framework, we must explicitly identify the increments corresponding to newly reported claims ($N_{i,j}$) and development changes on reported claims ($D_{i,j}$). Note that, in our setup, newly reported claims are immediately subject to the same continuous cost evolution (drift and diffusion) as previously reported ones.

For any square-integrable random variable $X^{\circ}\in\mathcal{F}_{s}^{i}$ with $0 \leq s \leq n$, we define $C_{t}^{i, \circ}[s, X^{\circ}]$ as the unique solution to the stochastic differential equation :

\begin{equation}\label{Ccirc} 
    C_{t}^{i, \circ}[s, X^{\circ}] = X^{\circ} - \int_{s}^{t} \delta_{u} C_u^{i, \circ}du + \int_{s}^{t} \tau_{u} \sqrt{C_u^{i, \circ}}dW_u^i, \quad s \leq t \leq n,
\end{equation}
which is the jump-free version of the original dynamics \eqref{EDS_C}.

\begin{lemme}\label{cND}
    Define the sequences $(N_{i,j}, D_{i,j})_{1 \leq i, j \leq n}$ by
    \[
        D_{j+1}^{i} := C_{j}^i - C_{j+1}^{i, \circ}[j, C_{j}^i], \quad N_{j+1}^{i} := C_{j+1}^{i} - C_{j+1}^{i, \circ}[j, C_{j}^i].
    \]
This yields the decomposition:
	\[
		C_{j+1}^{i} = C_{j}^{i} + N_{j+1}^{i} - D_{j+1}^{i}
	\]
The two first conditional moments are as follows.\\
For the $D$-component: 
         \[
    \begin{aligned}
        \mathbb{E}\left(D_{j+1}^i \mid \mathcal{F}_j^i\right) = \, &C_{j}^i \left(1-e^{\int_{j}^{j+1} -\delta_u du }\right),\\
        Var\left(D_{j+1}^i \mid \mathcal{F}_j^i\right) = \, &C_j^i\int_{j}^{j+1}\tau_u^2e^{-\int_{j}^u \delta_z dz - \int_{u}^{j+1} 2\delta_z dz}du.
    \end{aligned}
\]
For the $N$-component: 
         \[
    \begin{aligned}
        \mathbb{E}\left(N_{j+1}^i \mid \mathcal{F}_j^i\right) = \, &\mathbb{E}(Z)E_i\int_{j}^{j+1}\lambda_ue^{\int_{u}^{j+1} -\delta_v dv }du,\\
        Var\left(N_{j+1}^i \mid \mathcal{F}_j^i\right) = \,
    &\mathbb{E}(Z)E_i\int_{j}^{j+1}\tau_u^2e^{-\int_{u}^{j+1} 2\delta_v dv}\int_{j}^{u}\lambda_v e^{\int_{v}^{u}-\delta_y dy}dvdu \\
    &+ \mathbb{E}(Z^2)E_i\int_{j}^{j+1}\lambda_ue^{-2\int_{u}^{j+1}\delta_v dv}du.
    \end{aligned}
\].
 \end{lemme}
 \begin{proof}
     The conditional moments of $D_{j+1}^{i}$ conditional on $\mathcal{F}_{j}^{i}$ follow directly from \cite[Proposition 3.6]{baradel2025ime}, applied to the pure-diffusion process starting at $C_j^i$. For $N_{j+1}^i$, note that
\[
N_{j+1}^i = C_{j+1}^i - C_j^i + D_{j+1}^i,
\]
thus together with Proposition \ref{Cmoments}, we get the two first moments of $N_{j+1}^{i}$ condition on $\mathcal{F}_{j}^{i}$.
 \end{proof}

\begin{corollaire}\label{C_param}
The processes $(C_{t}^i)_{0 \leq t \leq n}$, $N_{i,j}$, $D_{i,j}$ for $1 \leq i \leq n$ satisfy \textit{H1}, \textit{H2} and \textit{H3} of Assumption \ref{H_schnieper} by setting, for $1 \leq j \leq n-1$:
\[
    \begin{aligned}
    1-\Delta_j &:= e^{\int_{j}^{j+1}-\delta_udu},\\
    \Lambda_{j+1} &:= \mathbb{E}(Z)\int_{j}^{j+1}\lambda_ue^{\int_{u}^{j+1} -\delta_v dv }du,\\
    T_j^2 &:= \int_{j}^{j+1}\tau_u^2e^{-\int_{j}^u \delta_z dz - \int_{u}^{j+1}2\delta_z dz}du.\\
    \Sigma_{j+1}^2 &:= \mathbb{E}(Z)\int_{j}^{j+1}\tau_u^2e^{-\int_{u}^{j+1} 2\delta_v dv}\int_{j}^{u}\lambda_v e^{\int_{v}^{u}-\delta_y dy}dvdu \\
    &+ \mathbb{E}(Z^2)\int_{j}^{j+1}\lambda_ue^{-\int_{u}^{j+1}2\delta_v dv}du.
    \end{aligned}
\]
\end{corollaire}

We now derive a useful distributional representation for $C$ (and hence also for $N$ and $D$). We first recall the branching property associated with \eqref{Ccirc}.

\begin{remarque}
The solution to \eqref{Ccirc} satisfies the branching property: if $C'^{i, \circ}[s, X'^{\circ}]$ is an independent copy of the process satisfying \eqref{Ccirc}, then the sum $C^{i, \circ}[s, X^{\circ}] + C'^{i, \circ}[s, X'^{\circ}]$ also satisfies \eqref{Ccirc} with initial value $X^{\circ} + X'^{\circ}$ at time $s$ (driven by a different Brownian motion); see \cite[Lemma 3.3]{baradel2025ime}.
\end{remarque}

The key idea is as follows: at the $k$-th jump time $\tau_k$, the process jumps by $C_{\tau_k} = C_{\tau_k-} + Z_k$. Exploiting the branching property, we can decompose the post-jump process into the continuation of the pre-jump part (following \eqref{Ccirc}) and a new independent branch starting from $Z_k$ at time $\tau_k$ and evolving according to \eqref{Ccirc}. The next lemma formalizes this representation.

\begin{lemme}\label{lemme_S}
Let $0 \leq s < \tau_{\ell}^i \leq t$ denote the jump times of $\Pi^i$ on $[s, t]$, then
    \[
        C_{t}^{i} \mid \sigma(\mathcal{F}_{s}^{i}, \sigma(\Pi^i([s, t], \cdot))) \overset{law}{=} S^{i, 0}_t[s, C_s^i] + \sum_{\ell=1}^{\Pi_{t}^i-\Pi_{s}^i}S_{t}^{i, \ell}[\tau_{\ell}, Z_{\Pi_s^i + \ell}^i]
    \]
    where the processes $S^{i, \ell}$ are independent copies of the solution to \eqref{Ccirc}, each driven by its own independent Brownian motion.
\end{lemme}

\begin{proof}
 All statements are understood conditionally on $\sigma(\mathcal{F}_{j}^{i}, \sigma(\Pi^i([s, t], \cdot)))$. Under this conditioning, the jump times $\tau_\ell$ and sizes $Z_{\Pi_s^i + \ell}$ for $1 \leq \ell \leq \Pi_t^i - \Pi_s^i$ are fixed and known on $[s, t]$, , and they are independent of the Brownian increments $W_u^i - W_s^i$ for $u \geq s$.

    \smallbreak
    
    \noindent Since $C^{i}$ solves \eqref{EDS_C}, we have at each jump time $\tau_{\ell}$:
    \[
    C_{\tau_{\ell}}^i = C_{\tau_{\ell}-}^i + Z_{\Pi_{s}^{i} + \ell}^i.
    \]
    Set $\tau_0 := s$. We proceed by induction on $0 \leq k \leq \Pi_t^i - \Pi_s^i$. 

    \smallbreak
    
    \noindent \textbf{1.} Base case ($k=0$): For $s \leq t < \tau_1$, there are no jumps in $[s, t]$, so $C_t^i$ evolves according to \eqref{Ccirc} starting from $C_s^i$ at time $s$. Thus,
    \[
     C^i_{t} = S^{i, 0}_{t}[s, C_s^i].
    \]
    \textbf{2.} Inductive hypothesis: Assume that for some $k\geq 1$,        
        \[
        C^i_{t} \overset{law}{=} S^{i, 0}_{t}[s, C_s^i] + \sum_{\ell=1}^{k-1}S_{t}^{i, \ell}[\tau_{\ell}, Z_{\Pi_s^i + \ell}^i],\quad \tau_{k-1} \leq t < \tau_{k}.
        \]
    where the $S^{i, \ell}$ for $\ell \geq 0$ are independent copies satisfying \eqref{Ccirc}. In particular, at the end of this interval,
    \[
    C^i_{\tau_{k}-} \overset{law}{=} S^{i, 0}_{\tau_{k}-}[s, C_s^i] + \sum_{\ell=1}^{k-1}S_{\tau_{k}-}^{i, \ell}[\tau_{\ell}, Z_{\Pi_s^i + \ell}^i],
    \]
    \textbf{3.} Inductive step: Now consider $\tau_k \leq t < \tau_{k+1}$. At $\tau_k$, we add the jump: $C_{\tau_k}^i = C_{\tau_k-}^i + Z_{\Pi_s^i + k}$. From $\tau_k$ onward, the process evolves continuously according to \eqref{Ccirc} until the next jump. Note that at time $\tau_k$, the new branch $S_{\tau_{k}}^{i, k}[\tau_{k}, Z_{\Pi_s^i + k}^i] = Z_{\Pi_{s}^{i} + k}^i$. By the inductive hypothesis, the independence of the branches, and the branching property applied to the sum at $\tau_k-$ plus the new jump, the evolution from $\tau_k$ to $t$ of $C_t^i$ has the same conditional law as the sum of the continued evolutions of the existing branches plus the new independent branch starting at $Z_{\Pi_{s}^{i} + k}^i$:
    \[
        C^i_{t} \overset{law}{=} S^{i, 0}_{t}[s, C_s^i] + \sum_{\ell=1}^{k}S_{t}^{i, \ell}[\tau_{\ell}, Z_{\Pi_s^i + \ell}^i], \quad \tau_{k} \leq t < \tau_{k+1}.
    \]
This completes the induction.
\end{proof}

\begin{remarque}\label{rND}
Conditioning on $\mathcal{F}_j^i$ and using the definitions of $N$ and $D$ from Lemma \ref{cND}, we have the following equalities in law:
	\[
    \begin{aligned}
        N_{i, j+1} &\overset{law}{=} \sum_{\ell=1}^{\Pi_{j+1}^i-\Pi_{j}^i}S_{j+1}^{i, \ell}[\tau_{\ell}, Z_{\Pi_j^i + \ell}^i],\\
        D_{i, j+1} &\overset{law}{=} C_{i,j} - S^{i, 0}_{j+1}[j, C_j^i].
    \end{aligned}
    \]

\end{remarque}

Seasonal effects may be present within each year of development. Moreover, there is no need to track the functions $t \mapsto (\lambda_t, \delta_t, \tau_t)$ continuously over time. To simplify the model while retaining sufficient flexibility, we introduce the following piecewise-constant assumption.

    \begin{hypothese}\label{H_constant}
The functions $\lambda$, $\delta$, and $\tau$ are constant on each $[t, t+1)$, i.e., for $1 \leq t < n$:
    \[
            \lambda_t := \sum_{j=0}^{n-1}\lambda_{j+1}\mathbf{1}_{[j, j+1)}(t), \quad \delta_t := \sum_{j=0}^{n-1}\delta_j\mathbf{1}_{[j, j+1)}(t), \quad \tau_t := \sum_{j=0}^{n-1}\tau_j\mathbf{1}_{[j, j+1)}(t).
    \]
    \end{hypothese}

Under this assumption, the parametric relations from Corollary \ref{C_param} simplify considerably.

\begin{lemme}\label{small}
    Under Assumption \ref{H_constant}, the relation in Corollary \ref{C_param} for $0 \leq j \leq n-1$ simplifies to
    \[
        \begin{aligned}
            1-\Delta_j &= e^{-\delta_j} \\
            \Lambda_{j+1} &= \frac{\lambda_{j+1}\mathbb{E}(Z)}{\delta_j}\left(1-e^{-\delta_j}\right) \\
            T^2_j &= \frac{\tau_j^2}{\delta_j}\left(e^{-\delta_j} - e^{-2\delta_j}\right) \\
            \Sigma^2_{j+1} &= \frac{{\tau_j^2}\lambda_{j+1}\mathbb{E}(Z)(1-e^{-\delta_j})^2}{2\delta_j^2} + \frac{\lambda_{j+1}\mathbb{E}(Z^2)(1-e^{-2\delta_j})}{2\delta_j}
        \end{aligned} \Longleftrightarrow
        \begin{aligned}
            \delta_j &= -\log(1-\Delta_j) \\
            \lambda_{j+1} &= \frac{-\Lambda_{j+1}\log(1-\Delta_j)}{\mathbb{E}(Z)\Delta_j} \\
            \tau^2_j &= \frac{T_j^2 \log(1-\Delta_j)}{\Delta_j(\Delta_j-1)} \\
            \Sigma^2_{j+1} &= A_j
            + B_j\frac{\mathbb{E}(Z^2)}{\mathbb{E}(Z)}
        \end{aligned}
    \]
where 
\[
A_j := \frac{T_j^2\Lambda_{j+1}}{2(1-\Delta_j)}, \quad B_j := \frac{\Lambda_{j+1}(2-\Delta_j)}{2}.
\]
and with $\Delta_0 = T_0 = 0$ (hence $\delta_0 = \tau_0 = 0$).
\end{lemme}

\begin{proof}
The proof proceeds directly from Corollary \ref{C_param} by evaluating a series of straightforward integrals.
\end{proof}

\begin{remarque}\label{X_ratio}
For the \textit{true IBNR} part, the parameters $(\lambda_j)_{1 \leq j \leq n}$, $\mathbb{E}(Z)$, and $\mathbb{E}(Z^2)$ determine both the $\Lambda$'s and the $\Sigma$'s. In particular, the second-moment quantities $\Sigma_j^2$ admit the linear structure :
    \begin{equation}\label{reg}
         \Sigma^2_j = A_j + B_j \cdot X, \quad 1 \leq j \leq n-1,
    \end{equation}
where $A_j$ and $B_j$ are as defined in Lemma \ref{small} and $X := \frac{\mathbb{E}(Z^2)}{\mathbb{E}(Z)}$. Thus, $X$ can be estimated with a simple linear regression (with no constant), which fixes one degree of freedom for the distribution of $Z$, but does not impose it. Thus, since \eqref{reg} can be rewritten simply as:
    \[
         \Sigma^2_j - A_j =  X \cdot B_j , \quad 1 \leq j \leq n-1,
    \]
$X$ can be estimated with a simple linear regression with no intercept. Since $\Sigma_j$ is typically estimated using $n-j$ observations (degrees of freedom), we perform weighted least squares with weights $w_j = n - j$.
	\[
		\widehat{X} := \frac{\sum_{j=1}^{n-1} (n-j) B_j (\Sigma^2_j - A_j)}{\sum_{j=1}^{n-1} (n-j) B_j^2 }.
	\]
The variance of $Z$ can then be recovered as
	\[
		Var(Z) = \mathbb{E}(Z)(X - \mathbb{E}(Z)),
	\]
which implies, in particular, that
	\[
		0 < \mathbb{E}(Z) < X.
	\]
\end{remarque}

Note that in the stochastic differential equation \eqref{EDS_C}, as $C_t^i$ approaches zero, both the drift coefficient and the diffusion coefficient vanish.  In the absence of the $N$ component (i.e., no new claims arrivals) the process $C_{i,t}$ can be solved explicitly on each interval $j \leq t \leq j+1$ conditional on $\mathcal{F}_j^i$, for any realized value of $C_{i,j}$.

\begin{lemme}\label{PC0}
For $j \leq t \leq j+1$, under Assumption \ref{H_constant}, the following bounds hold:
	\[
		\exp\left(-\lambda_{j+1} E_i (t-j)-\frac{2\delta_{j}e^{-\delta_{j}(t-j)}C_j^i}{\tau_j^2\left(1 - e^{-\delta_{j}(t-j)}\right)}\right) \leq \mathbb{P}(C_{t}^i = 0 \mid \mathcal{F}_{j}^i) \leq \mathbb{P}(D_{t}^i = C_{j}^{i} \mid \mathcal{F}_{j}^i),
	\]
where the random variable $D_{t}^i$ is defined in Remark \ref{rND} and
\[
	\mathbb{P}(D_{t}^i = C_{j}^{i} \mid \mathcal{F}_{j}^i) = \exp\left(-\frac{2\delta_{j}e^{-\delta_{j}(t-j)}C_j^i}{\tau_j^2\left(1 - e^{-\delta_{j}(t-j)}\right)}\right).
\]
This shows that the process $(C_{t}^i)_{t \geq 1}$ can reach 0 (and stay at zero until the next IBNR claim occurs). In practice, however, this probability is often numerically very close to zero. Moreover, conditional on being positive, the distribution of $C_{t}^{i}$ is continuous.
\end{lemme}

\begin{proof}
Since $Z \geq 0$, we have $C_t \geq C_{t}^{i, \circ}[j, C_{j}^i]$ for $j \leq t \leq j+1$, where $C^{i, \circ}[j, C_{j}^i]$ satisfies equation \eqref{Ccirc}. The upper bound therefore follows directly from \cite[Remark 3.12]{baradel2025ime}.

\smallbreak

For the lower bound, note that
    \[
        \mathbb{P}(C_{t}^i = 0 \mid \mathcal{F}_{j}^i) \geq \mathbb{P}(\left\{C_{t}^i = 0\right\} \cap \left\{\Pi^i\left([j, t], \mathbb{R}_+^*\right) = 0\right\} \mid \mathcal{F}_{j}^i).
    \]
    combined, again, with \cite[Remark 3.12]{baradel2025ime}.
    \[
        \mathbb{P}(\left\{C_{t}^i = 0\right\} \cap \left\{\Pi^i\left([j, t], \mathbb{R}_+^*\right) = 0\right\} \mid \mathcal{F}_{j}^i) \leq \mathbb{P}(C_{t}^i = 0 \mid \mathcal{F}_{j}^i) \leq \mathbb{P}(C_{t}^{\circ, i}[j, C_{j}^{i}] = 0 \mid \mathcal{F}_{j}^i)
    \]
and 
    \[
    \mathbb{P}(\left\{C_{t}^i = 0\right\} \cap \left\{\Pi^i\left([j, t], \mathbb{R}_+^*\right) = 0\right\} = \exp\left(-\lambda_{j+1} E_i (t-j)\right)\mathbb{P}(C_{t}^{\circ, i}[j, C_{j}^{i}] = 0 \mid \mathcal{F}_{j}^i)
    \]
    From \cite[Remark 3.12]{baradel2025ime},
    \[
        \mathbb{P}(C_{t}^{\circ, i}[j, C_{j}^{i}] = 0 \mid \mathcal{F}_{j}^i) = \exp\left(-\frac{2\delta_{j}e^{-\delta_{j}(t-j)}C_j^i}{\tau_j^2\left(1 - e^{-\delta_{j}(t-j)}\right)}\right),
    \]
    which leads to the result.
\end{proof}

Unlike in \cite{baradel2025ime}, the distribution of the jump size $Z$ is only constrained by the fixed ratio of its first two moments (see Remark \ref{X_ratio}). As a result, we cannot in general obtain a closed-form expression for the full conditional distribution of $C_t^i$. Nevertheless, exact simulation remains possible.

\begin{proposition}\label{loiS}
    Let $(z_{\ell})_{\ell \geq 0}$ and $(t_{\ell})_{\ell \geq 0}$ be two sequences of non-negative real numbers. For $j \leq t \leq j+1$, define $(\lambda_{\ell})_{\ell \geq 0}$ and $(\beta_{\ell})_{\ell \geq 0}$ by:
    	\[
		\lambda_\ell := \frac{-2\delta_{j}e^{-\delta_{j}(t-t_{\ell})}z_{\ell}}{\tau_j^2\left(e^{-\delta_{j}(t-t_{\ell})} - 1\right)}, \quad \beta_\ell := \frac{-2\delta_{j}}{\tau_j^2\left(e^{-\delta_{j}(t-t_{\ell})} - 1\right)}.
	\]
    For each $\ell \geq 0$, let $S_{\ell}[t_{\ell}, z_{\ell}]$ be the independent random variable
    	\[
    	S_t^{i, \ell}[t_{\ell}, z_{\ell}] := \sum_{k=1}^{M_{\ell}} X_k^{\ell},
    	\]
where $M_{\ell} \sim \mathcal{P}(\lambda_\ell)$ and $(X_k^{\ell})_{k \geq 1} \overset{i.i.d.}{\sim} \mathcal{E}(\beta_{\ell}) $, independent of the Poisson process $\Pi$.

\smallbreak

Then, for $j \leq t \leq j+1$
    \[
        C_{t}^{i} \mid \sigma(\mathcal{F}_{j}^{i}, \sigma(\Pi^i([j, t], \cdot)) \overset{law}{=} S_t^{i,0}[j, C_j^i] + \sum_{\ell=1}^{\Pi_{t}^i-\Pi_{j}^i}S_t^{i, \ell}[\tau_{\ell}, Z_{\Pi_j^i + \ell}^i]
    \]
    where $j < \tau_{\ell}^i \leq t$ are the jump times of $\Pi^i$ in the interval $[j, t]$.
\end{proposition}
\begin{proof}
From Lemma \ref{cND} we already have the decomposition involving the independent processes $(S^{i, \ell})_{\ell \geq 0}$. Each $S^{i, \ell}$ satisfies the same SDE as in Lemma \ref{cND} (started at time $t_\ell$ from level $z_\ell$), and its conditional law at time $t$ is given explicitly in \cite[Lemma 3.13]{baradel2025ime}.
\end{proof}

\begin{remarque}\label{loiSr}
    To simulate $C_{j+1}^{i}$ conditionally on $C_{j}^{i}$, two different strategies can be considered. In both approaches, the jump times and jump sizes $(\tau_{\ell}, Z_{\ell})_{1 \leq \ell \leq \Pi_{j+1}^i - \Pi_j^i}$ are first simulated.

    \smallbreak

    \noindent \textbf{Method 1.} Simulate the process jump by jump following the SDE \eqref{EDS_C}:
    \[
    \begin{aligned}
        C_{\tau_\ell-}^{i} &= C_{\tau_\ell}^{i, \circ}[\tau_{\ell-1}, C_{\tau_{\ell-1}}^i] \\
        C_{\tau_\ell}^{i} &= C_{\tau_\ell-}^{i} + Z_{\ell},
    \end{aligned}
    \]
    \smallbreak

    \noindent \textbf{Method 2.} Use the explicit representation given in Proposition \eqref{loiS}:
    \[
     C_{j+1}^{i} = S_{j+1}^{i,0}[j, C_j^i] + \sum_{\ell=1}^{\Pi_{j+1}^i-\Pi_{j}^i}S_{j+1}^{i, \ell}[\tau_{\ell}, Z_{\Pi_j^i + \ell}^i]
    \]
    the sequences $(N_{j}^{i})$ and $(D_{j}^{i})$ can be identified using Remark \ref{rND}.
\end{remarque}

In Remark \ref{loiSr}, the first method does not give the $(N_{j}^i)$ and $(D_{j}^i)$ but only the $C_j^i$.

\bigbreak

It remains to specify the distribution of the jump sizes $Z$. Observe that, in Lemma \ref{small}, the parameters $(\lambda_j)$ depend on $\mathbb{E}[Z]$. With only aggregated data (as is typically the case in classical reserving triangles), no further information about $Z$ is available, so the choice of its distribution remains partially free (subject to moment constraints). 

\smallbreak

When individual claim data are available, one can additionally exploit the number of new claims occurring in each development year $j$ to estimate the $(\lambda_j)$ parameters separately for each development period. The distribution of $Z$ can then be estimated from the amounts of all new claims across all development years (i.e., treating $Z$ as independent of the development delay $j$). 

\smallbreak

From Remark \ref{X_ratio}, the ratio $\frac{\mathbb{E}[Z^2]}{\mathbb{E}[Z]}$ is already identifiable and can be estimated from the data. We now provide a concrete example using a gamma distribution for $Z$, which will later serve to illustrate the proposed bootstrap methodology.

\begin{exemple}\label{gamma}
Assume $Z \sim \mathcal{G}(\alpha, \beta)$. Let $\widehat{X}$ be an estimator of $\frac{\mathbb{E}[Z^2]}{\mathbb{E}[Z]}$ obtained from Remark \ref{X_ratio}. After fixing (or estimating) $\mathbb{E}[Z]$, the first two moments of the gamma distribution can be matched by setting
\[
	\widehat{\alpha} = \frac{\mathbb{E}(Z)}{\widehat{X} - \mathbb{E}(Z)}, \quad \widehat{\beta} = \frac{1}{\widehat{X} - \mathbb{E}(Z)},
\]
provided that $0 < \mathbb{E}(Z) < \widehat{X}$.
\end{exemple}

\section{The bootstrap methodology}\label{C_bootstrap}

The bootstrap procedure consists of two main classical steps:

\begin{enumerate}
    \item Bootstrapping the model parameters ($\Lambda_j$, $\Delta_j$, $T_j$, $\Sigma_j$) to account for \emph{estimation error};
    \item Simulating the unobserved future increments in the lower part of the triangle using the bootstrapped parameters, thereby incorporating \emph{process error}.
\end{enumerate}
We adapt the bootstrap approach described in \cite{england2006predictive} to the present continuous-time setting. We denote by $\mathbb{P}_{Z}$ the true distribution of $Z$, and by $\widehat{\mathbb{P}}_{Z}$ its estimated version (as illustrated, for example, in Example \ref{gamma}).

\medbreak

\noindent \textbf{1 - Bootstrapping the parameters.} For each $(i, j)$ such that $i + j \leq  n$,
we simulate $(N_{j+1}^{i, m}, D_{j+1}^{i, m})_{1 \leq m \leq M}$ conditionally on the observed $C_j^i$, using Remark \ref{loiSr} with the estimated parameters $(\widehat{\lambda}_{j+1}, \widehat{\delta}_{j}, \widehat{\tau}_{j})_{1 \leq j \leq n-1}$ and the distribution $\widehat{\mathbb{P}}_Z$.

\medbreak

We obtain the bootstrapped estimators $(\widehat{\Lambda}_{j}^m, (\widehat{\Sigma}_{j}^m))$ and $(\widehat{\Delta}_{j}^m, \widehat{T}_{j}^m)$ defined as, for all $1 \leq m \leq M$:

\begin{equation}
    \begin{aligned}
    \widehat{\Delta}_{j}^m &:= \frac{\sum_{i = 1}^{n - j}D_{j+1}^{i, m}}{\sum_{i = 1}^{n - j}C_{j}^i}, & 1 \leq j \leq n-1, \\
    (\widehat{T}_{j}^{m})^2 &:= \frac{1}{n-j-1}\sum_{i = 1}^{n-j}C_{j}^i\left(\frac{D_{j+1}^{i, m}}{C_{j}^i} - \widehat{\Delta}_{j}^m\right)^{2}, & 1 \leq j \leq n-2,\\
    \widehat{\Lambda}_{j}^m &:= \frac{\sum_{i = 1}^{n - j+1}N_{j}^{i, m}}{\sum_{i = 1}^{n - j+1}E_{i}}, & 1 \leq j \leq n, \\
    (\widehat{\Sigma}_{j}^{m})^2 &:= \frac{1}{n-j}\sum_{i = 1}^{n-j+1}E_{i}\left(\frac{N_{j}^{i, m}}{E_{i}} - \widehat{\Lambda}_{j}^m\right)^{2}, & 1 \leq j \leq n-1.
    \end{aligned}
\end{equation}
We then derive $(\widehat{\lambda}_{j}^m, \widehat{\delta}_{j}^m, \widehat{\tau}_{j}^m)$ using Lemma \ref{small}, and $\widehat{\mathbb{P}}_Z^m$ using Remark \ref{X_ratio} and, e.g., Example \ref{gamma}. 

\noindent \textbf{2 -  Bootstrapping the process error.} For each $(i, j)$ such that $i+j > n+1$, for all $i > 1$, we simulate $C_{j}^{i, m}$ using Remark \ref{loiSr} with the bootstrapped parameters $(\widehat{\lambda}_{j}^m, \widehat{\delta}_{j}^m, \widehat{\tau}_{j}^m)$ and the bootstraped distribution $\widehat{\mathbb{P}}_Z^m$ iteratively on $n-i+1 \leq j \leq n$ starting from $C^{i, m}_{n-i+1} := C^{i}_{n-i+1}$.

\noindent \textbf{3 - The bootstrapped reserve.} This yields the bootstrapped simulations of the reserves:
\begin{equation}\label{r_bootstrap}
	R^m := \sum_{i=2}^{n} C_{n}^{i, m} - C_{n-i+1}^{i}, \quad 1 \leq m \leq M.
\end{equation}
The vector $(R^m)_{1 \leq m \leq M}$ approximates the distribution of the reserves, conditional on the observed data.

\section{Numerical example}

We implement our bootstrap method within a continuous-time framework and apply it to the dataset from \cite{schnieper1991separating}, which originates from a real-world motor third-party liability excess-of-loss pricing problem. The data include the triangles $N$ and $D$ shown in Table~\ref{triangleND}, as well as the exposure values and the triangle $C$ presented in Table~\ref{triangleEC}.

        \begin{table}[H]
        \begin{center}\setlength{\tabcolsep}{1.4mm}
        \fontsize{9pt}{12pt}\selectfont
        \begin{tabular}{|c|c c c c c c c|}
  \hline
  $i$ & $N_{i,1}$ & $N_{i,2}$ & $N_{i,3}$ & $N_{i,4}$ & $N_{i,5}$ & $N_{i,6}$ & $N_{i,7}$ \\
  \hline
1 & 7.5 & 18.3 & 28.5 & 23.4 & 18.6 & 0.7 & 5.1\\
2 & 1.6 & 12.6 & 18.2 & 16.1 & 14 & 10.6 &\\
3 & 13.8 & 22.7 & 4 & 12.4 & 12.1 & &\\
4 & 2.9 & 9.7 & 16.4 & 11.6 & & &\\
5 & 2.9 & 6.9 & 37.1 & & & &\\
6 & 1.9 & 27.5 & & & & &\\
7 & 19.1 & & & & & &\\
  \hline
\end{tabular}\quad \quad \quad\begin{tabular}{|c|c c c c c c c|}
  \hline
  $i$ & $D_{i,1}$ & $D_{i,2}$ & $D_{i,3}$ & $D_{i,4}$ & $D_{i,5}$ & $D_{i,6}$ & $D_{i,7}$ \\
  \hline
1 & 0 & -3.1 & 4.8 & -8.5 & 23 & 3.9 & 2.5 \\
2 & 0 & -0.6 & 0.9 & 8.6 & -1.4 & 5.6 & \\
3 & 0 & -5.9 & 10.1 & -4.6 & -31.1 & & \\
4 & 0 & -1.4 & -2.1 & -2.8 & & & \\
5 & 0 & 0 & -5.8 & & & & \\
6 & 0 & 0 & & & & & \\
7 & 0 & & & & & & \\
  \hline
\end{tabular}
\end{center}
\caption{The $N$ (left) and $D$ (right) triangles from the dataset in \cite{schnieper1991separating}.\label{triangleND}}
\end{table}

        \begin{table}[H]
        \begin{center}\setlength{\tabcolsep}{1.4mm}
        \fontsize{9pt}{12pt}\selectfont
        \begin{tabular}{|c|c|}
  \hline
  $i$ & $E_{i}$ \\
  \hline
1 & 10224\\
2 & 12752\\
3 & 14875\\
4 & 17365\\
5 & 19410\\
6 & 17617\\
7 & 18129\\
  \hline
\end{tabular}\quad \quad \quad\begin{tabular}{|c|c c c c c c c|}
  \hline
  $i$ & $C_{i,1}$ & $C_{i,2}$ & $C_{i,3}$ & $C_{i,4}$ & $C_{i,5}$ & $C_{i,6}$ & $C_{i,7}$ \\
  \hline
1 & 7.5 & 28.9 & 52.6 & 84.5 & 80.1 & 76.9 & 79.5\\
2 & 1.6 & 14.8 & 32.1 & 39.6 & 55 & 60 &\\
3 & 13.8 & 42.4 & 36.3 & 53.3 & 96.5 & &\\
4 & 2.9 & 14 & 32.5 & 46.9 & & &\\
5 & 2.9 & 9.8 & 52.7 & & & &\\
6 & 1.9 & 29.4 & & & & &\\
7 & 19.1 & & & & & &\\
  \hline
\end{tabular}
\end{center}
\caption{The exposure $E$ (left) and cumulative claims triangle $C$ (right) from the dataset in \cite{schnieper1991separating}.\label{triangleEC}}
\end{table}

To identify the distribution of $Z$, we first estimate the ratio $\frac{\mathbb{E}(Z^2)}{\mathbb{E}(Z)}$. Following Remark~\ref{X_ratio}, this quantity is estimated via a weighted linear regression through the origin (i.e., with no intercept term). We regress $(\Sigma_j^2 - A_j)$ on $(B_j)$ with the noise $(\varepsilon_j)$ :

\[
\Sigma_j^2 - A_j = B_j \times \frac{\mathbb{E}(Z^2)}{\mathbb{E}(Z)} + \varepsilon_j,
\]
using the pairs $(A_j, B_j)$ obtained from Lemma~\ref{small}, and with weights $w_j := n - j$.

Figure~\ref{fig:regs} illustrates this regression. The left panel (Figure~\ref{fig:reg}) shows the observed points $(B_j, \Sigma_j^2 - A_j)$ together with the fitted line. The right panel (Figure~\ref{fig:sreg}) compares the original $\Sigma_j$ values with the estimated values $\widehat{\Sigma}_j$ obtained from the regression (note that $\Sigma_n := 0$ and receives zero weight in the regression). Table~\ref{tab:reg} reports the estimation results.\\

\begin{figure}[H]
    \centering
    
    \begin{subfigure}[b]{0.48\textwidth}
        \centering
        \includegraphics[width=\textwidth]{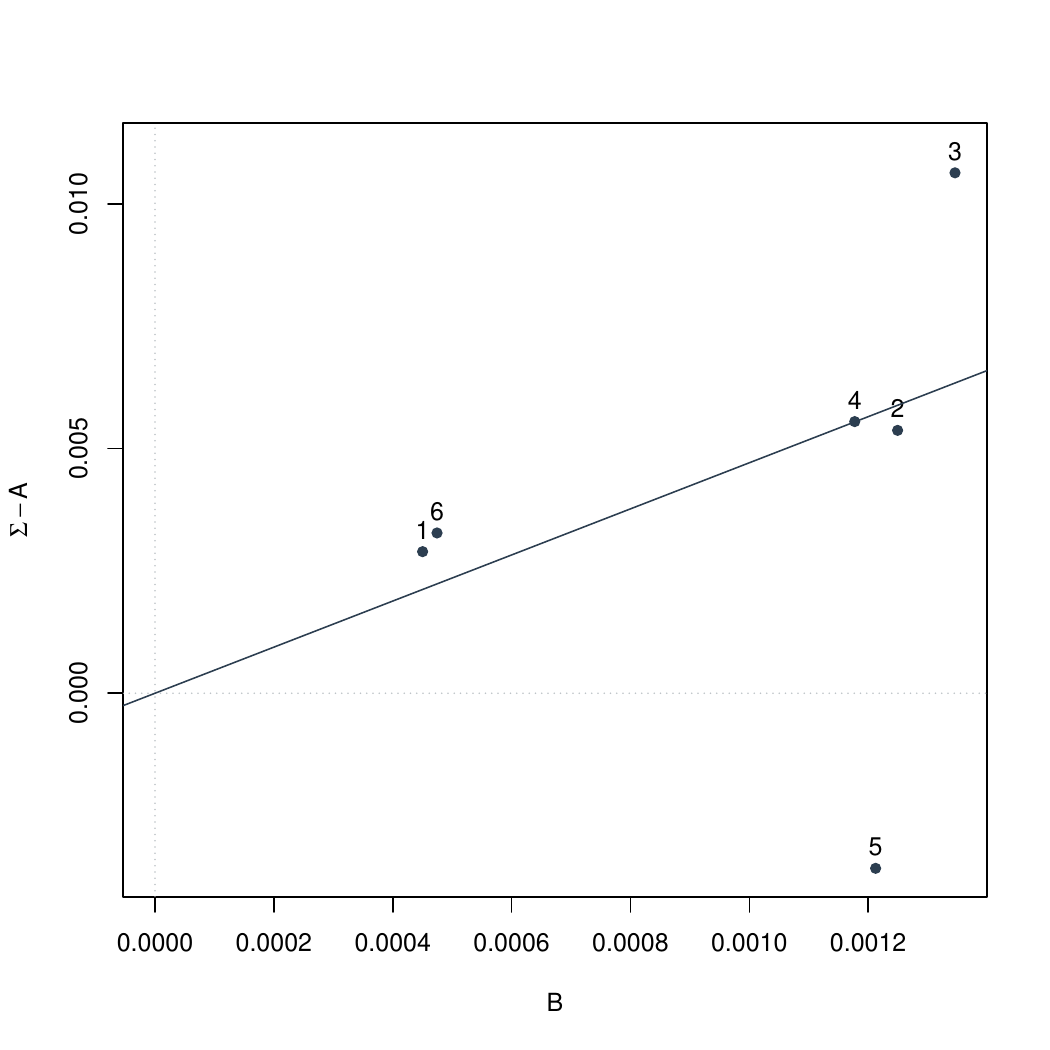}
        \caption{$\Sigma_j^2 - A_j$ versus $B_j$ with indices $j$}
        \label{fig:reg}
    \end{subfigure}
    \hfill 
    \begin{subfigure}[b]{0.48\textwidth}
        \centering
        \includegraphics[width=\textwidth]{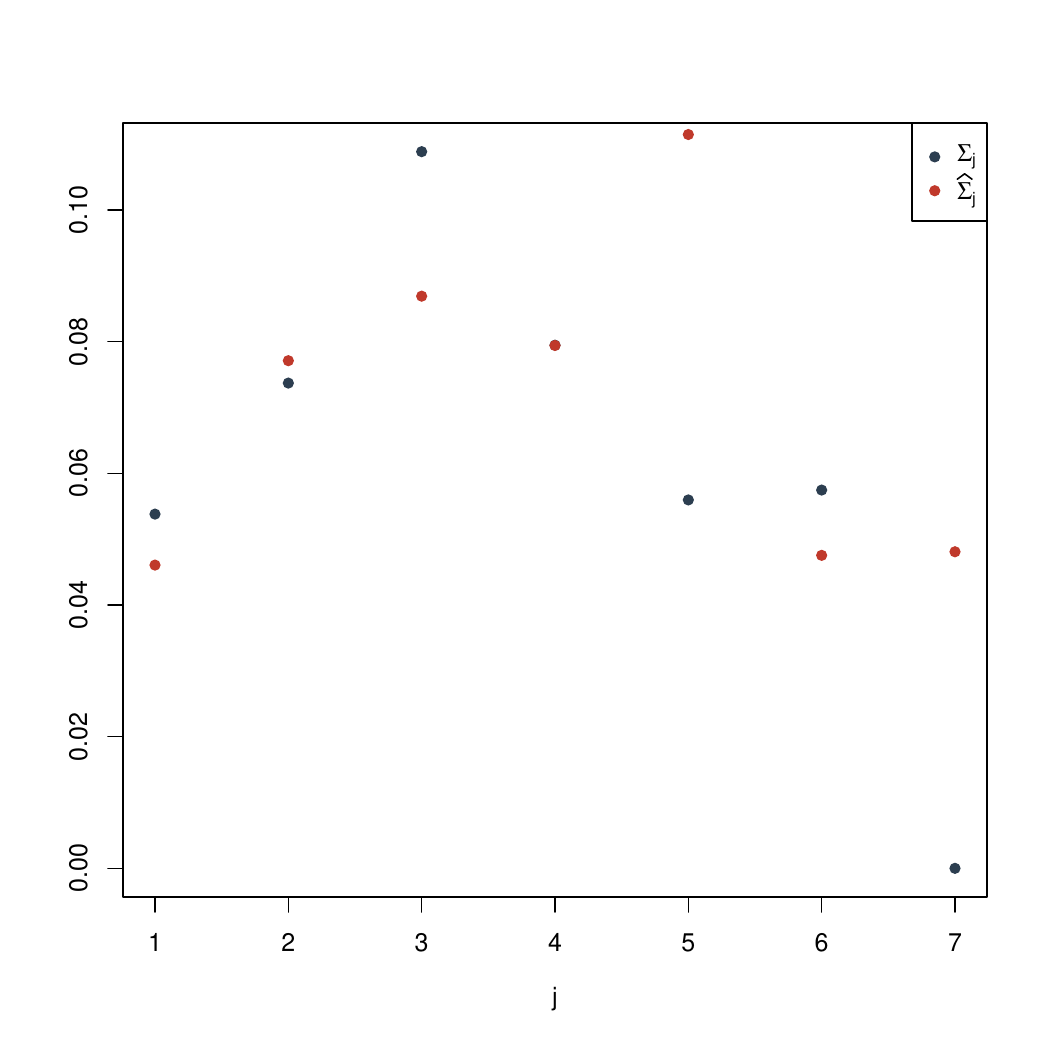}
        \caption{Comparison between $(\Sigma_j)$ and $(\widehat{\Sigma}_j)$}
        \label{fig:sreg}
    \end{subfigure}
    
    \caption{Illustration of the weighted linear regression for $\Sigma$}
    \label{fig:regs}
\end{figure}

\begin{table}[H]
        \begin{center}\setlength{\tabcolsep}{1.4mm}
        \fontsize{9pt}{12pt}\selectfont
        \begin{tabular}{|c|c|c|}
  \hline
  $\widehat{\frac{\mathbb{E}(Z^2)}{\mathbb{E}(Z)}}$ & $p$-value & $R^2$ \\
  \hline
4.7120 & 0.0235 & 0.6747\\
  \hline
\end{tabular}
\end{center}
\caption{Estimation of $\frac{\mathbb{E}(Z^2)}{\mathbb{E}(Z)}$ following Remark \ref{X_ratio}.\label{tab:reg}}
\end{table}

From Lemma \ref{small}, we compute the $(\lambda_{j}\mathbb{E}[Z])$ for each development period $j$ presented in Table \ref{tab:lambdaEZ}. These quantities are independent of the value of $\mathbb{E}[Z]$ itself.

\begin{table}[H]
        \begin{center}\setlength{\tabcolsep}{1.4mm}
        \fontsize{9pt}{12pt}\selectfont
        \begin{tabular}{|c|c|}
  \hline
  $j$ & $\lambda_{j}\mathbb{E}[Z]$  \\
  \hline
1 & $0.4502954 \times 10^{-3}$\\
2 & $0.9048361 \times 10^{-3}$\\
3 & $1.4490241 \times 10^{-3}$\\
4 & $1.1235202 \times 10^{-3}$\\
5 & $1.1504111 \times 10^{-3}$\\
6 & $0.5099654 \times 10^{-3}$\\
7 & $0.5071148 \times 10^{-3}$\\
\hline
\end{tabular}
\end{center}
\caption{Estimation of $\lambda_{j}\mathbb{E}(Z)$ which is independent of $\mathbb{E}(Z)$, following Lemma \ref{small}.\label{tab:lambdaEZ}}
\end{table}

No additional information is provided in Schnieper's data that would allow us to directly determine a specific value for $\mathbb{E}(Z)$. Once $\mathbb{E}(Z)$ is fixed, the implied average annual number of new claims ranges from a minimum of $\frac{4.604}{\mathbb{E}(Z)}$ to a maximum of $\frac{28.13}{\mathbb{E}(Z)}$. Using Table~\ref{tab:reg}, the variance of $Z$ can be expressed as

	\[
		Var(Z) = \mathbb{E}[Z]\left(\widehat{\frac{\mathbb{E}[Z^2]}{\mathbb{E}[Z]}} - \mathbb{E}[Z]\right) = \mathbb{E}[Z](4.7120 - \mathbb{E}[Z]).
	\]
Moreover, we must also specify a full distribution for $Z$. We choose the gamma distribution, the parameters are then obtained as follows:
    \[
        Z \sim \mathcal{G}(\widehat{\alpha}, \widehat{\beta}), \quad \widehat{\beta} = \frac{1}{4.7120 - \mathbb{E}(Z)}, \quad \widehat{\alpha} = \mathbb{E}(Z)\widehat{\beta}.
    \]

\subsection{Schnieper's model with a parameterized Log-normal distribution}

We employ the classic estimator
	\[
\widehat{C}_{i,n} := \left(\prod_{k = n+1-i}^{n-1}(1-\widehat{\Delta}_{k})\right)C_{i,n+1-i} + E_{i}\sum_{k = n+2-i}^{n}\widehat{\Lambda}_{k}\left(\prod_{\ell = k}^{n-1}(1-\widehat{\Delta}_{\ell})\right),
\]
which leads to the estimation of the expected value for the total reserve:
    \[
     \widehat{\mu}_{R} := \widehat{R}.
    \]
We denote by $\widehat{\sigma}^{2}_{R}$ the conditional MSEP of \cite{liu2009predictive}. Finally, we approximate the distribution of the reserve with a Log-normal distribution by matching the moments:
    \[
        \mathcal{LN}\left(\log(\widehat{\mu}_{R}) - \frac12 \log\left(1+\frac{\widehat{\sigma}^{2}_{R}}{\widehat{\mu}_{R}}\right) , \log\left(1+\frac{\widehat{\sigma}^{2}_{R}}{\widehat{\mu}_{R}}\right)\right).
    \]

\subsection{Schnieper's model with Bootstrap}\label{M_bootstrap}

For comparison purposes, we calculate both the conditional MSEP and its distribution using a bootstrap method, where the conditional MSEP is the variance of the bootstrap distribution.

\medbreak

We employ the procedure outlined in \cite{liu2009bootstrap}, which we briefly summarize here. First, we compute the Pearson residuals $(r_{i,j})$ from the triangle $N$ and $(s_{i,j})$ from the triangle $D$.

\medbreak

\noindent \textbf{1.} We simulate:
    \[
        \begin{aligned}
        N_{i,j}^m &:= \widehat{\Lambda}_{j}E_{i} + \widehat{\Sigma}_{j}\sqrt{E_{i}}r_{i,j}^m, \quad &1 \leq m \leq M,\\
        D_{i,j+1}^m &:= \widehat{\Delta}_{j}C_{i,j} + \widehat{T}_{j}\sqrt{C_{i,j}}s_{i,j}^m, \quad &1 \leq m \leq M,
        \end{aligned}
    \]
    where each $r_{i,j}^m$ and $s_{i,j}^m$ is drawn independently and uniformly from the set of observed Pearson residuals $(r_{i,j})$ and $(s_{i,j})$. Using \eqref{est_fs}, we compute $(\widehat{\Lambda}_{j+1}^m, \widehat{\Sigma}_{j+1}^m, \widehat{\Delta}_{j}^m, \widehat{T}_{j}^m)_{1 \leq j \leq n-1}$ for $1 \leq m \leq M$.

\medbreak

\noindent \textbf{2.}  We initiate the simulation with $C_{i,n-i+1}^{m} := C_{i,n-i+1}$, and then iteratively simulate the lower triangle for $2 \leq i \leq n$ as follows:
        \[
            C_{i,j+1}^m \sim \mathcal{N}\left(\widehat{\Lambda}_{j+1}^{m}E_{i} + (1-\widehat{\Delta}_{j}^{m})C_{i,j}^m,\, (\widehat{\Sigma}_{j+1}^{m})^{2}E_{i} + (\widehat{T}_{j}^{m})^{2}C_{i,j}^m\right),
        \]
        and we deduce the bootstrap distribution of the total reserve with the formula (\ref{r_bootstrap}).

\subsection{Time series with Bootstrap}

Following the idea of \cite{buchwalder2006mean} for Mack's model, we propose a Time series model that fits Schnieper's assumption. 

\begin{equation*}
\begin{aligned}
    N_{i,j} &=  \zeta_{i,j}, \\
    D_{i, j+1} &= \Delta_{j} C_{i,j} + T_{j} \sqrt{C_{i,j}} \varepsilon_{i,j}
    \end{aligned}
\end{equation*}
where the $(\varepsilon_{i,j})$ are centered with unit variance and $\zeta_{i,j}$ are independent such that $\mathbb{E}(\zeta_{i,j}) = \lambda_j E_i$ and $Var(\zeta_{i,j}) = \sigma_j^2 E_i$. Thus, $N$ is a simple independent process while $D$ is a simple Time series. We introduce the following hypothesis:
\[
    \begin{aligned}
	(\varepsilon_{i,j})_{1 \leq i, j \leq n} &\overset{i.i.d.}{\sim} \mathcal{N}(0, 1),\\
	(\zeta_{i,j})_{1 \leq i, j \leq n} & \sim \mathcal{G}\left(\frac{\Lambda_j^2}{\Sigma_j^2}E_i, \frac{\Lambda_j}{\Sigma_j^2}\right).
    \end{aligned}
\]
Now, we describe the bootstrap method for this model.

\medbreak

\noindent \textbf{1.} For $i + j \leq n + 1$, we simulate the upper triangle:
    \[
    \begin{aligned}
        N_{i,j}^m &\sim \mathcal{G}\left(\frac{\widehat{\Lambda}_j^2}{\widehat{\Sigma}_j^2}E_i, \frac{\widehat{\Lambda}_j}{\widehat{\Sigma}_j^2}\right), \quad &1 \leq m \leq M,\\
        D_{i,j+1}^m &\sim \mathcal{N}\left(\widehat{\Delta}_{j}C_{i,j}, \widehat{T}_{j}^{2}C_{i,j}\right), \quad &1 \leq m \leq M.
    \end{aligned}
    \]
Using \eqref{est_fs}, we derive $(\widehat{\Lambda}_{j+1}^m, \widehat{\Sigma}_{j+1}^m, \widehat{\Delta}_{j}^m, \widehat{T}_{j}^m)_{1 \leq j \leq n-1}$ for $1 \leq m \leq M$. Additionally, it is noteworthy that when $C_{i,j}$ is fixed,

\begin{equation}
    \begin{aligned}
    \widehat{\Delta}_{j}^m &\sim \mathcal{N}\left(\widehat{\Delta}_{j}, \frac{\widehat{T}^2_j}{\sum_{i=1}^{n-j}C_{i,j}}\right), \\
    (\widehat{T}_{j}^{m})^2 &\sim \widehat{T}_j^2\frac{\chi^2_{n-j-1}}{n-j-1},
    \end{aligned}
\end{equation}
and both are independent. We can simulate the $(\widehat{\Delta}_{j}^m, \widehat{T}_{j}^{m})_{1 \leq j \leq n-1}$ directly. We can also see that $\widehat{\Lambda}_{j}^m \sim \mathcal{G}\left(\frac{\widehat{\Lambda}_j^2}{\widehat{\Sigma}_j^2}\sum_{i=1}^{n-j+1}E_i, \frac{\widehat{\Lambda}_j}{\widehat{\Sigma}_j^2}\sum_{i=1}^{n-j+1}E_i\right)$ but we do not easily have the distribution of $\widehat{\Sigma}_j^m$, we need to simulate the $N_{i,j}^m$ in order to get them.

\medbreak

\noindent \textbf{2.} As in Section \ref{M_bootstrap}, we begin with $C_{i,n-i+1}^{m} := C_{i,n-i+1}$, we simulate iteratively, for $2 \leq i \leq n$ and $i + j > n+1$ the lower triangle:
        \[
            \begin{aligned}
            N_{i,j}^m \sim \mathcal{G}\left(\frac{(\widehat{\Lambda}_j^m)^2}{(\widehat{\Sigma}_j^m)^2}E_i, \frac{\widehat{\Lambda}_j^m}{(\widehat{\Sigma}_j^m)^2}\right),
            D_{i,j+1}^m \sim \mathcal{N}\left(\widehat{\Delta}_{j}^{m}C_{i,j}^m, (\widehat{T}_{j}^{m})^{2}C_{i,j}^m\right),
            \end{aligned}
        \]
        and we deduce the bootstrap distribution of the total reserve with the formula \eqref{r_bootstrap}.

\subsection{Comparison and conclusion}

Table~\ref{table_bootstrap} compares the conditional prediction of several methods, measured by the square root of the estimated conditional MSEP and the estimated 99.5\% quantile excess ($Q(R; 99.5\%) - \widehat{R}$), both expressed as percentages of the point estimate of the reserve $\widehat{R}$.

\begin{table}[H]
        \begin{center}\footnotesize
        \begin{tabular}{|c|c|c|}
  \hline
   Method & $ \sqrt{\widehat{MSEP}}$ (in \% of $\widehat{R}$) & $Q(R; 99.5\%) - \widehat{R}$ (in \% of $\widehat{R}$)  \\
  \hline
  Schnieper Log-normal & 42.9175 & 165.003 \\
  Schnieper Bootstrap & 38.1737 & 103.181 \\ 
  Time series Bootstrap & 37.1173 & 114.056 \\
  Continuous-time Bootstrap ($\mathbb{E}[Z] = 1$) & 43.1650 & 136.702\\

  \hline
\end{tabular}
\end{center}
\caption{Conditional MSEP and 99.5\% quantile excess for the three other methods introduced and the continuous-time bootstrap from Section \ref{C_bootstrap}.}\label{table_bootstrap}
\end{table}

However, adjusting the \textit{Schnieper Log-normal} to employ a Gamma distribution reduces the quantile to $144.387\%$, highlighting the critical influence of the chosen distribution.

\medbreak

In simulations, the Schnieper Bootstrap often yields negative values for $N_{i, j}$, occurring roughly in $66\%$ of the simulations, while $D_{i, j+1} > C_{i,j}$ never occurs. The Time Series Bootstrap method has by construction $N_{i, j} \geq 0$ but leads to $D_{i, j+1} > C_{i,j}$ in $6\%$ of the simulations.

\medbreak

The processes $(C_t^i)_{0 \leq t \leq n}$ remain non-negative in the continuous-time model and, in particular, $D_{j+1}^{i} \leq C_{j}^{i}$ conditional on $\mathcal{F}_{j}^i$, for all $(i, j)$. Nevertheless, in Lemma \ref{PC0} , we noted that $\mathbb{P}(D_{j+1}^{i} = C_{j}^{i} \mid \mathcal{F}_{j}^i) > 0$, on the diagonal, the highest probability arises for $j = 4$, and we have:
\[
	\mathbb{P}(D_{5}^{n-3} = C_{4}^{n-3} \mid \mathcal{F}_{4}^{n-3}) \approx 2.604 \times 10^{-4}.
\]

In Figure~\ref{fig_bootstrap}, we display the complete distributions associated with the various models and our framework.

\begin{figure}[H]
\centering
\includegraphics[scale=0.595]{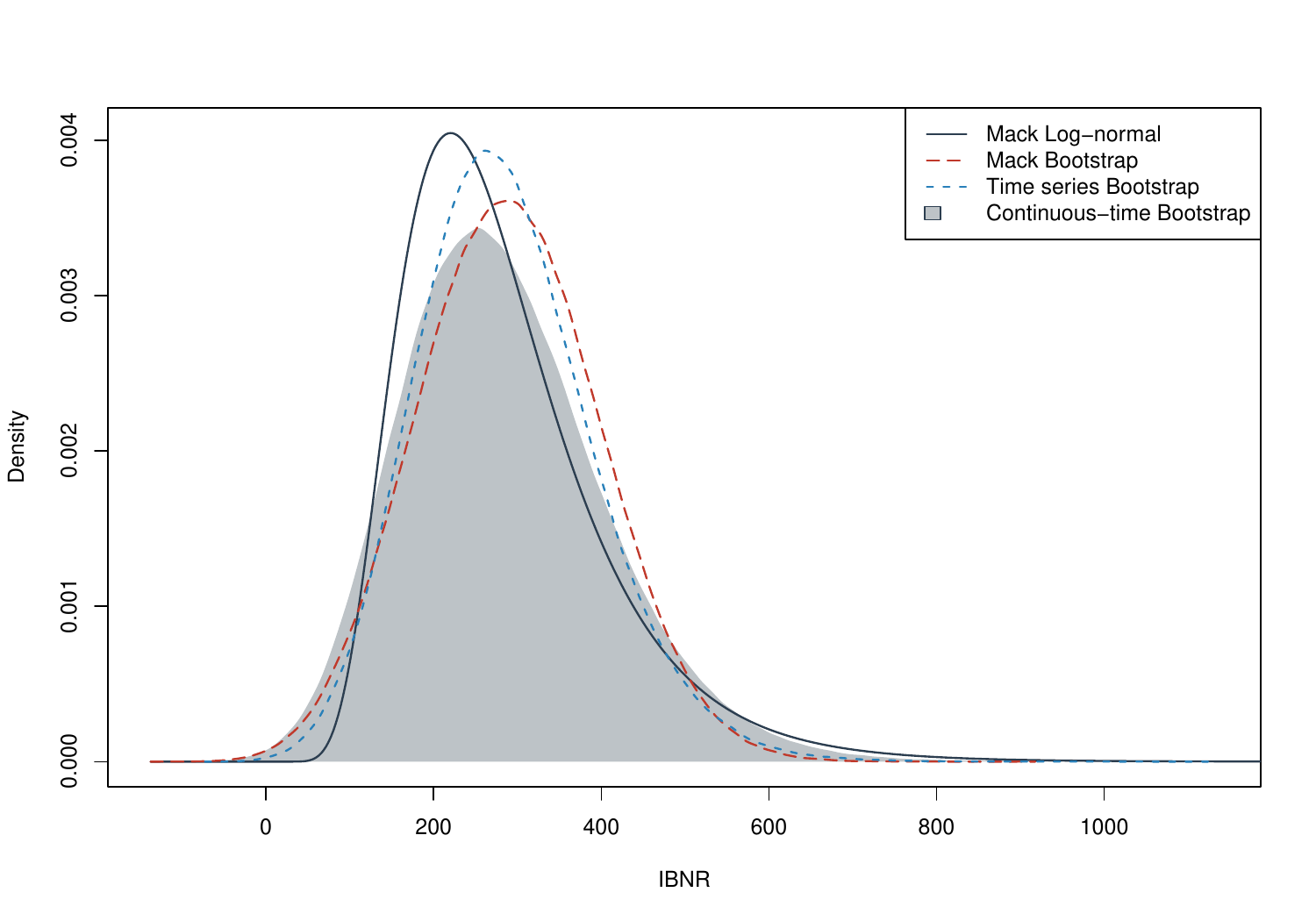}
\caption{Estimated conditional distributions of the total reserve.\label{fig_bootstrap}}
\end{figure}

\smallbreak

We now investigate the sensitivity of the results to the choice of $\mathbb{E}[Z]$, where $0 < \mathbb{E}[Z] < \widehat{\mathbb{E}[Z^2]/\mathbb{E}[Z]} \approx 4.712$. In the main analysis above, $\mathbb{E}[Z] = 1$ was used. Figure~\ref{fig:ez} displays $\sqrt{\widehat{\mathrm{MSEP}}}$ (left) and the 99.5\% quantile excess (right) as functions of $\mathbb{E}[Z]$. Both quantities exhibit noticeable sensitivity, particularly for smaller values of $\mathbb{E}[Z]$.

\smallbreak

This behaviour can be understood from the properties of the Gamma distribution. Lower values of $\mathbb{E}[Z]$ correspond to smaller shape ($\alpha$) and rate ($\beta$) parameters. The exponential moments of the Gamma distribution exist only up to order $\beta$, and the excess kurtosis is $6/\alpha$. As $\alpha$ and $\beta$ increase (i.e., as $\mathbb{E}[Z]$ becomes larger), the distribution becomes less heavy-tailed and approaches normality, which affects both the MSEP and the tail quantiles.

\smallbreak

In the absence of additional information, a natural default choice for $\mathbb{E}[Z]$ is a value that approximately recovers the conditional MSEP of the original Schnieper model. When individual claim data are available, they provide a direct way to estimate $\mathbb{E}[Z]$ and, more generally, the full distribution $\mathbb{P}_Z$.

\begin{figure}[H]
    \centering
    
    \begin{subfigure}[b]{0.48\textwidth}
        \centering
        \includegraphics[width=\textwidth]{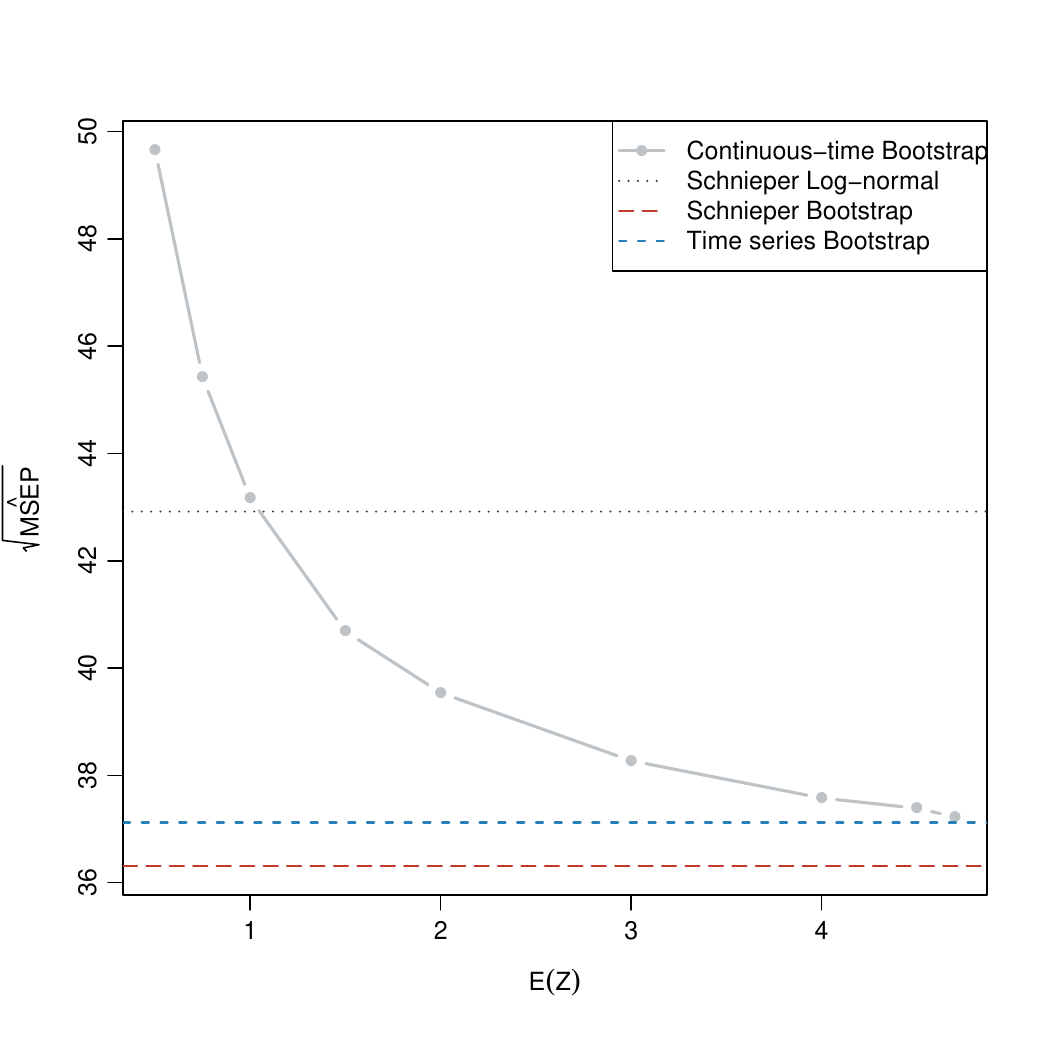}
        \caption{$\sqrt{\widehat{MSEP}}$}
        \label{fig:ez_mesp}
    \end{subfigure}
    \hfill
    \begin{subfigure}[b]{0.48\textwidth}
        \centering
        \includegraphics[width=\textwidth]{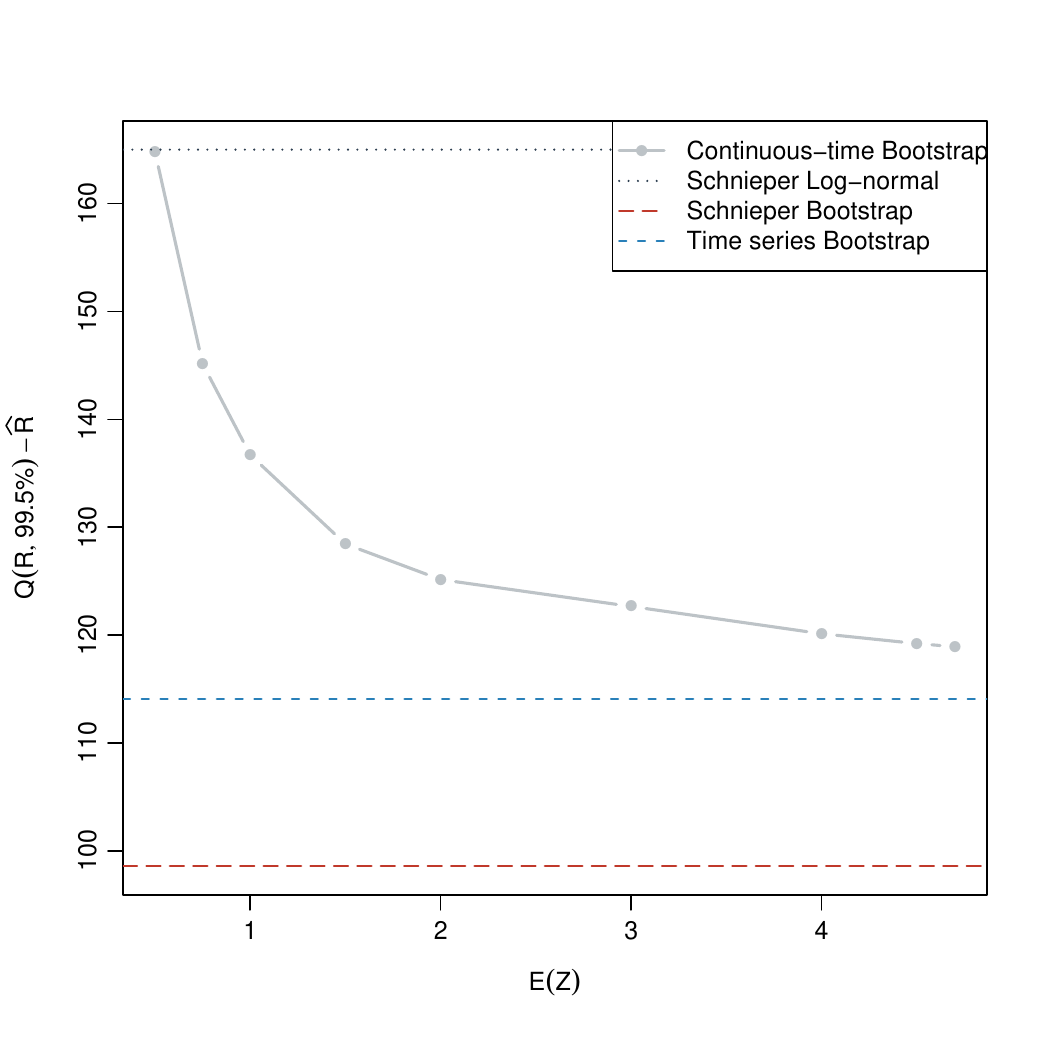}
        \caption{$Q(R; 99.5\%) - \widehat{R}$}
        \label{fig:ez_q}
    \end{subfigure}
    
    \caption{Comparison of $ \sqrt{\widehat{MSEP}}$ and $Q(R; 99.5\%) - \widehat{R}$ as a function of $\mathbb{E}(Z)$.}
    \label{fig:ez}
\end{figure}

\FloatBarrier

\section*{Acknowledgments}

The author acknowledges the financial support provided by the \emph{Fondation Natixis}.

\bibliographystyle{plain}
\bibliography{bibliographie}

\end{document}